\documentclass[smallextended]{svjour3}       
\usepackage[T1]{fontenc}
\usepackage{txfonts}

\smartqed  
\newcommand{\ket}[1]{|#1\rangle}
\newcommand{\bra}[1]{\langle #1 |}

\journalname{Quantum Information Processing}
\begin{document}
\hypersetup{pdftitle={Strongly Secure Quantum Ramp Secret Sharing Constructed from Algebraic Curves over Finite Fields},pdfauthor={Ryutaroh Matsumoto},pdfkeywords={algebraic curve, quantum secret sharing, non-perfect secret sharing, ramp secret sharing, strong security}}
\title{Strongly Secure Quantum Ramp Secret Sharing
Constructed from Algebraic Curves over Finite Fields}
\titlerunning{Quantum Ramp Secret Sharing by Algebraic Curves}
\author{Ryutaroh Matsumoto}
\institute{Ryutaroh Matsumoto \at
              Department of Communications and Computer Engineering,
              Tokyo Institute of Technology, Japan\\
and Department of Mathematical Sciences,
Aalborg University, Denmark\\
              ORCID: 0000-0002-5085-8879 \\
              \email{ryutaroh@it.ce.titech.ac.jp}}
\date{October 20, 2014}
\maketitle
\begin{abstract}
The first construction of
strongly secure quantum ramp secret sharing
by Zhang and Matsumoto
had an undesirable feature that the dimension of quantum shares must be
larger than the number of shares.
By using algebraic curves over finite fields,
we propose a new construction in which
the number of shares can become arbitrarily large for fixed
dimension of shares.
\keywords{algebraic curve \and
quantum secret sharing \and non-perfect secret sharing \and ramp secret sharing \and strong security}
\PACS{03.67.Dd}
\subclass{81P94 \and 94A62 \and 94B27}
\end{abstract}

\section{Introduction}
Secret sharing (SS) scheme encodes
a secret into multiple shares being
distributed to participants,
so that
only qualified sets of shares can reconstruct the secret
perfectly \cite{stinson06}.
The secret and shares are traditionally classical
information \cite{stinson06},
but now quantum secret and quantum shares can
also be used \cite{cleve99,gottesman00,smith00}.

In perfect SS, if a set of shares is not qualified,
that is, it cannot reconstruct the secret perfectly,
then the set has absolutely no information about the secret.
It is well-known that the share sizes in perfect SS
must be larger than or equal to that of the secret,
both in classical and quantum cases.
To overcome this inefficiency of storing shares,
the ramp classical SS was proposed \cite{blakley85,mceliece81,yamamoto86},
which reduces the share sizes at the cost of
allowing partial information leakage to non-qualified sets
of shares. In ramp SS, a share set is said to be forbidden
if it has no information about secret,
while it is said to be intermediate if it is neither
qualified nor forbidden \cite{iwamoto06,yamamoto86}.

The first quantum ramp SS was proposed by Ogawa et al.\ \cite{ogawa05},
which made the share size $L$ times smaller than its secret,
where $L$ is the number of qudits in the secret.
In their study \cite{ogawa05},
there were two drawbacks.
Firstly, it does not control how
information is leaked to a non-qualified set of shares,
and there exists an undesirable case in which
an intermediate set of shares can understand a qudit 
in the secret, as demonstrated in \cite{zhang14}.
To exclude such a possibility, we introduced a notion
of the strong security of quantum ramp SS, which ensures
no intermediate set can understand a qudit in the secret
(see \cite{zhang14} for its formal definition) and
proposed an explicit construction with the strong security.

The second drawback of \cite{ogawa05} as well as
our previous proposal \cite{zhang14} is that the dimension of quantum shares
must be larger than that of the number of participants.
When the number of participants is large,
handling quantum shares become more difficult,
because handling large dimensional quantum systems
are generally more difficult than smaller ones.
Our previous proposal \cite{zhang14} solved the first drawback
but did not the second.
The purpose of this paper is to solve the first and the second
drawbacks of \cite{ogawa05}
simultaneously.

We will proceed as follows:
Firstly, we modify the strong security definition
given in \cite{zhang14} in Section \ref{sec2},
because the previous definition in \cite{zhang14}
required that all the qualified sets are of the same size,
and also that all the forbidden sets are of the same size.
Secondly, in Section \ref{sec3},
we carry over the classical
strongly secure ramp SS \cite{chen08,matsumoto14css}
using algebraic curves to the quantum setting, then
we prove that the proposed quantum SS has the strong security.
We also present sufficient conditions
for its qualified, intermediate, and forbidden sets
by using the technique in \cite{matsumoto14qss}.
We conclude this paper in Section \ref{sec4}.

\section{Extended Definition of the Strong Security}\label{sec2}
Let $q$ be a prime power, 
$\mathcal{G}_i$ ($i=1$, \ldots, $L$) and
$\mathcal{H}_j$ ($j=1$, \ldots, $n$) be the $q$-dimensional
complex linear spaces,
where $\mathcal{G}_i$ contains the $i$-th qudit of the quantum
secret, while $\mathcal{H}_j$ contains the $j$-th quantum share.
$L$ is the number of qudits in secret and $n$ is the number
of shares or participants.
In this paper we consider the so-called pure state scheme
\cite{cleve99,gottesman00},
in which a pure state secret is converted to pure state shares.
Encoding is an isometric complex linear map from
$\bigotimes_{i=1}^L \mathcal{G}_i$ to
$\bigotimes_{j=1}^n \mathcal{H}_j$.
A subset $J \subset \{1$, \ldots, $n\}$
is said to be qualified if the quantum secret is
perfectly reconstructed from the aggregated shares in
$\bigotimes_{j\in J} \mathcal{H}_j$,
forbidden if the aggregated shares in
$\bigotimes_{j\in J} \mathcal{H}_j$ is always the same
quantum state regardless of the quantum secret,
and intermediate otherwise,
as defined in \cite{ogawa05}.

We introduce a new definition of the strong security,
which does not require the qualified and the forbidden sets
being the same size.
Let $I \subseteq \{1$, \ldots, $L\}$,
$J \subseteq \{1$, \ldots, $n\}$,
$\overline{I} = \{1$, \ldots, $L\}
\setminus I$, and
$\overline{J} = \{1$, \ldots, $n\}
\setminus J$.
Define $\mathcal{G}_I = \bigotimes_{i \in I} \mathcal{G}_i$,
and $\mathcal{G}_{\overline{I}} = \bigotimes_{i \in \overline{I}} \mathcal{G}_i$.
The idea behind the following strong security with respect to
$I$ and $J$ is that the share set $J$ has no idea on
what is a quantum state $\rho_I$
on the part $\mathcal{G}_I$ of the quantum secret.
To formally express this idea, the quantum state $\sigma_J$ of shares
on
$\bigotimes_{j \in J} \mathcal{H}_j$ is required to be independent of
$\rho_I$.
On the other hand, $\sigma_J$ also depends on the quantum
state on $\mathcal{G}_{\overline{I}}$.
When an illegitimate owner of the shares in $J$
is guessing $\rho_I$, she or he is assumed to have
no prior knowledge on the part $\mathcal{G}_{\overline{I}}$,
which enables us to use the fully mixed state as the state
on $\mathcal{G}_{\overline{I}}$.

By using the above ideas, we formally define our extended
version of the strong security.
\begin{definition}\label{def10}
We retain notations from the above discussion.
A quantum ramp secret sharing scheme is said to be strongly secure
with respect to $I$ and $J$
if the quantum state $\sigma_J$ on the share set $J$
is always the same state regardless of the quantum
state $\rho_I \otimes \rho_{\overline{I},\mathrm{mix}}$
of the whole quantum secret,
where $\rho_{\overline{I},\mathrm{mix}}$ is the fully mixed state
on $\mathcal{G}_{\overline{I}}$.
\end{definition}
In our previous paper \cite{zhang14},
a $(k,L,n)$ quantum ramp SS (in the sense of \cite{ogawa05})
was said to be strongly secure
if all $I$ and $J$ with $|I|+|J| \leq k$ satisfy
Definition \ref{def10}, where $k$ was the minimum size of share sets
which can perfectly reconstruct the secret, and $L,n$ had the
same meaning as the present paper.

\section{Explicit Construction of Strongly Secure Quantum Ramp SS}\label{sec3}
In the previous constructions \cite{ogawa05,zhang14}
of quantum ramp SS, shares are generated by using
evaluations of a polynomial at pairwise distinct numbers
in the finite field $\mathbf{F}_q$ with $q$ elements.
Obviously $q$ must be larger than $n$ in those constructions.
In the above constructions, the dimension of quantum shares
is also $q$, and larger values of $q$ usually make
implementation difficult.
The restriction $q > n$ also exists in the classical SS
based on evaluations of a polynomial \cite{mceliece81,shamir79}.
One of standard ways in classical SS
to overcome the restriction $q>n$ is to use
points on an algebraic curve as done in \cite{chen08}.
We will propose an explicit strongly secure quantum ramp SS
based on the idea in \cite{chen08}.

It is well-known that an algebraic curve is
mathematically equivalent to an algebraic function field
of one variable \cite{bn:stichtenoth}.
So we will describe our proposal by using terminology
of algebraic function fields, as done in \cite{chen08}.
We briefly review the algebraic function fields,
see \cite{bn:stichtenoth} for a formal exposition.
The rational function field $\mathbf{F}_q(x)$
over $\mathbf{F}_q$ is the set of $f(x)/g(x)$,
where $f(x)$ and $g(x)$ are polynomials in $x$ with their
coefficients in $\mathbf{F}_q$.
Addition, subtraction, multiplication and division in
$\mathbf{F}_q(x)$ are defined in the standard way.
An algebraic function field $F$ is an extension
field of $\mathbf{F}_q(x)$ such that
the dimension of $F$ as an $\mathbf{F}_q(x)$-linear space
is finite. It is usually denoted as $F/\mathbf{F}_q$ to
indicate that it is defined by equations over $\mathbf{F}_q$.

\begin{example}\label{ex1}
Let $F$ be the field obtained by
adding $y$ to $\mathbf{F}_4(x)$, where
$y$ is a root of the univariate polynomial
$y^2 + y = x^3$ ($x^3$ is regarded as a coefficient).
Then $F$ is an algebraic function field of one variable,
and denoted by $\mathbf{F}_4(x,y)$.
The process of creating $\mathbf{F}_4(x,y)$ from $\mathbf{F}_4(x)$
is the same in spirit as creating the field of complex numbers
from that of real numbers by adding a root of $z^2 = -1$.

Observe also that the equation $y^2+y = x^3$
can also be seen as an algebraic curve.
There are eight points $R_1$,
\ldots, $R_8 \in \mathbf{F}_4^2$
satisfying $y^2+y = x^3$.
For example, $(x,y)=(0,1)$ satisfies $y^2+y = x^3$
and can be $R_1$.
Those eight points can be used for evaluations
in the SS proposed in \cite{chen08} and also in
our proposal described later.
Note that usable points for evaluation
increase from $4$ to $8$.

In the following
we will use so-called $\mathbf{F}_q$-rational
places. $R_1$, \ldots, $R_8$ are examples
of $\mathbf{F}_4$-rational places in this
function field.
The solutions of the defining equation
of $F$, e.g.\ $y^2+y=x^3$, are a subset
of $\mathbf{F}_q$-rational
places, provided that the curve defined by the equation
is \emph{smooth}. See \cite{bn:stichtenoth}
for formal definitions.
\end{example}

We return to the general description of our proposal.
Let $P_1$, \ldots, $P_n$, $Q_1$, \ldots, $Q_L$
be pairwise distinct $\mathbf{F}_q$-rational places
of $F/\mathbf{F}_q$.
A divisor of $F/\mathbf{F}_q$ is a formal sum
of (not necessarily $\mathbf{F}_q$-rational)
places $F/\mathbf{F}_q$, e.g.\ $2R_1 - R_3$ in Example \ref{ex1}.
The support of a divisor $G$ is the set of places whose
coefficient in $G$ is nonzero.
For example, the support of $2R_1 - R_3$ is
the set $\{R_1$, $R_3\}$.
Let $G$ be a divisor whose support contains
none of $P_1$, \ldots, $P_n$, $Q_1$, \ldots, $Q_L$.
For any divisor $G$,
there is a finite-dimensional
$\mathbf{F}_q$-linear space $\mathcal{L}(G)$,
see \cite{bn:stichtenoth} for a formal definition.
\begin{example}\label{ex2}
Consider again $\mathbf{F}_4(x,y)/\mathbf{F}_4$
introduced in Example \ref{ex1}.
Let $Q$ be the common pole of $x$ and $y$,
in other words, the unique point at infinity
belonging to the projective algebraic curve defined by $y^2+y=x^3$.
Then a basis of $\mathcal{L}(uQ)$ as
an $\mathbf{F}_4$-linear space  is
\begin{equation}
\{ x^a y^b \mid 0\leq a, 0\leq b \leq 1, 2a+3b \leq u \}. \label{eq2}
\end{equation}
Thus, an element $h \in \mathcal{L}(uQ)$
is a polynomial in which every term is a multiple of a
monomial in (\ref{eq2}).
We can obtain a value in $\mathbf{F}_4$
by substituting $x,y$ in $h$ by components in $R_i$ (for example
$R_1=(0,1)$)
defined in Example \ref{ex1}.
The obtained value is called the evaluation of $h$ at $R_i$
and denoted by $h(R_i)$.

In our proposal as well as \cite{chen08},
we also use another linear space
$\mathcal{L}(G-Q_1 - \cdots - Q_L)$.
When $G=uQ$ as above, we have
\[
\mathcal{L}(G-Q_1 - \cdots - Q_L) = \{
h \in \mathcal{L}(G) \mid h(Q_i) = 0, \mbox{ for }
i=1, \ldots, L\}.
\]
\end{example}

Now we are ready to describe our proposal.
Since we assumed $\dim \mathcal{G}_i = \dim \mathcal{H}_j = q$
for all $i,j$, we can assume their orthonormal basis to be
$\{ \ket{a} \mid a \in \mathbf{F}_q \}$. Then
the basis of $\bigotimes_{i=1}^L \mathcal{G}_i $
can be written as $\{ \ket{\vec{s}} \mid \vec{s} \in \mathbf{F}_q^L \}$.
To describe quantum ramp SS,
it is sufficient to specify the quantum state of shares
corresponding to a quantum secret
$\ket{\vec{s}} \in \bigotimes_{i=1}^L \mathcal{G}_i $
for every $\vec{s} \in \mathbf{F}_q^L$ as done in \cite{ogawa05,zhang14}.
We assume that
\begin{eqnarray}
L &=& \dim \mathcal{L}(G) - \dim \mathcal{L}(G-Q_1 - \cdots - Q_L),
\label{eq3}\\
0&=& \dim \mathcal{L}(G-P_1 - \cdots - P_n). \label{eq5}
\end{eqnarray}
The secret $\ket{\vec{s}}$ is encoded to
\begin{equation}
\frac{1}{\sqrt{q^{\dim \mathcal{L}(G - Q_1 - \cdots - Q_L)}}}
\sum_{\begin{array}{c}\scriptstyle h \in \mathcal{L}(G)\\
\scriptstyle (h(Q_1),\ldots, h(Q_L)) = \vec{s}
\end{array}}
\ket{h(P_1)} \otimes \ket{h(P_2)} \otimes \cdots \otimes \ket{h(P_n)}.
\label{eq4}
\end{equation}
The mapping $h \in \mathcal{L}(G)$
to $(h(Q_1)$, \ldots, $h(Q_L))$ is $\mathbf{F}_q$-linear
and its kernel is $\mathcal{L}(G - Q_1 - \cdots - Q_L)$
(see the end of Example \ref{ex2}).
By (\ref{eq3}) this mapping is surjective,
and for any $\vec{s} \in \mathbf{F}_q^L$ there exist
$q^{\dim \mathcal{L}(G - Q_1 - \cdots - Q_L)}$ elements $h \in \mathcal{L}(G)$
satisfying $(h(Q_1)$,\ldots, $h(Q_L)) = \vec{s}$,
which justifies the normalization factor in (\ref{eq4}).

On the other hand,
(\ref{eq5}) ensures that
the mapping $h \in \mathcal{L}(G)$
to $(h(P_1)$, \ldots, $h(P_n))$ is $\mathbf{F}_q$-linear and
injective. This guarantees that
terms appearing in the summation of (\ref{eq4}) do not overlap
for different $\vec{s}$, $\vec{s}' \in \mathbf{F}_q^L$,
which means that the encoded shares (\ref{eq4}) for
different $\vec{s}$, $\vec{s}'$ are orthogonal to each other.
From these discussions we see that (\ref{eq4})
maps an orthonormal basis to a subset of an orthonormal basis,
and that (\ref{eq4}) defines a complex linear isometric embedding,
from $\bigotimes_{i=1}^L \mathcal{G}_i$ to
$\bigotimes_{j=1}^n \mathcal{H}_j$.

\begin{remark}
One of the two classical ramp SS proposed by
Chen et al.\ \cite{chen08} is as follows:
For a classical secret
$(s_1$, \ldots, $s_L) \in \mathbf{F}_q^L$,
an element $h \in \mathcal{L}(G)$ with
$h(Q_i) = s_i$ ($i=1$, \ldots, $L$) is chosen uniformly
randomly. Then the $j$-th share is computed
as $h(P_j)$. Its similarity to our proposal
(\ref{eq4}) should be
obvious.
\end{remark}

For its strong security, we have the following theorem.
\begin{theorem}\label{th1}
The above quantum ramp SS is strongly secure with
respect to $I \subset \{1$, \ldots, $L\}$ and
$J \subset \{1$, \ldots, $n\}$ if
\begin{equation}
|J| \leq |\overline{I}| + \min\{ \deg G - L- 2g(F) + 1, n-1-\deg G\},
\label{eq6}
\end{equation}
where $g(F)$ denotes the genus of the algebraic function field
$F/\mathbf{F}_q$, see \cite{bn:stichtenoth} for its definition.
\end{theorem}
Its proof is technically complicated and heavily uses
the theory of algebraic function field \cite{bn:stichtenoth},
so we move it to Appendix \ref{app1}.

For quantum ramp SS to be useful,
a procedure for reconstructing the quantum secret
and sufficient conditions for qualified and forbidden sets
are indispensable.
On the other hand, actually the above proposal is a special
case of quantum ramp SS constructed from algebraic curves
studied in \cite{matsumoto14qss}.
By straightforward application of \cite{matsumoto14qss},
$\{1$, \ldots, $n\} \supset J$ is
qualified if
\begin{equation}
|J| \geq \max\{1+\deg G, n - (\deg G - L - 2g(F) + 1)\}, \label{eq7}
\end{equation}
and $J$ is forbidden if
\begin{equation}
|J| \leq  \min\{ \deg G - L- 2g(F) + 1, n-1-\deg G\}.
\label{eq8}
\end{equation}
Note that (\ref{eq6}) contains (\ref{eq8}) as its special case $|\overline{I}|=0$.
The reconstruction procedure in \cite{matsumoto14qss}
can also be used for the proposal in this paper.

The algebraic function field in Examples \ref{ex1} and \ref{ex2}
has genus $g(F)=1$. To make $n$ larger with fixed $q$,
we must find an algebraic function fields with many $\mathbf{F}_q$-rational
places. It is well-known \cite{bn:stichtenoth}
that the number of $\mathbf{F}_q$-rational
places is at most $1 + q + g(F) \lfloor 2 \sqrt{q} \rfloor$.
$F/\mathbf{F}_4$ in Examples \ref{ex1} and \ref{ex2} reaches this upper bound,
because the place $Q$ in Example \ref{ex2} is also $\mathbf{F}_4$-rational
and $F/\mathbf{F}_4$ in Examples \ref{ex1} and \ref{ex2} has
nine $\mathbf{F}_4$-rational places.
Requiring more $\mathbf{F}_q$-rational
places generally makes $g(F)$ larger, which makes
inequalities (\ref{eq6}), (\ref{eq7}) and (\ref{eq8}) weaker.
For fixed $q$ and $n$, it is desirable to use an algebraic
function field with smaller $g(F)$. Search for such ones has
been an active research area in pure mathematics
for past 30 years, see
\cite{bn:stichtenoth}.
In particular, it is known that
for fixed $q$ we can construct an algebraic function field
with arbitrarily many $\mathbf{F}_q$-rational
places.

\section{Conclusion}\label{sec4}
In this paper we argued that
the previously proposed strongly secure quantum ramp SS \cite{zhang14}
becomes difficult in implementation when the number $n$
of participants is large, because the dimension $q$
of
each quantum share must be $>n$.
To overcome this drawback,
we proposed new quantum ramp SS that allows arbitrarily
large $n$ for fixed $q$ while retaining the strong security.
The proposed construction is similar to the classical ramp SS
proposed by Chen et al.\ \cite{chen08}.

\begin{acknowledgements}
This research is partly supported by the National
Institute of Information and Communications Technology,
Japan, and by the Japan Society for the Promotion of Science Grant
 Nos.\ 23246071 and 26289116,
and the                                                                   
 Villum Foundation through their VELUX Visiting Professor Programme             
 2013--2014.
\end{acknowledgements}

\appendix
\section{Proof of Theorem \ref{th1}}\label{app1}
To prove Theorem \ref{th1},
we will prove a proposition covering a more general
class of quantum ramp SS.
We consider a quantum ramp SS constructed from
a pair of linear codes $C_2 \subsetneq C_1 \subseteq \mathbf{F}_q^n$
with $\dim C_1 - \dim C_2 = L$,
which was considered in \cite{matsumoto14qss}.

Encoding
is done as follows:
We will encode a quantum secret to
$n$ qudits in $\bigotimes_{j=1}^n \mathcal{H}_j$
by a complex linear isometric embedding.
To specify such an embedding, it is enough to specify
the image of each basis state $\ket{\vec{s}} \in \bigotimes_{i=1}^{L} \mathcal{G}_i$.
Fix an $\mathbf{F}_q$-linear isomorphism $f: \mathbf{F}_q^{\dim C_1 - \dim C_2}
\rightarrow C_1 / C_2$.
We encode $\ket{\vec{s}}$ to
\begin{equation}
\frac{1}{\sqrt{|C_2|}} \sum_{\vec{x} \in f(\vec{s})} \ket{\vec{x}} \in
\bigotimes_{j=1}^n \mathcal{H}_j. \label{eq10}
\end{equation}
Recall that by definition of $f$, $f(\vec{s})$ is a subset of $C_1$,
$f(\vec{s}) \cap f(\vec{s}_1) = \emptyset$ if $\vec{s} \neq \vec{s}_1$,
and $f(\vec{s})$ contains $|C_2|$ vectors. From these properties
we see that (\ref{eq10}) defines a complex linear isometric embedding.
The quantum system $\mathcal{H}_j$ is distributed to the $j$-th
participant.
For $I \subset \{1$, \ldots, $L\}$,
the map $P_I$ denotes the projection of a vector to the index set
$I$, that is, for $\vec{s}=(s_1$, \ldots, $s_L) \in \mathbf{F}_q^L$,
$P_I(\vec{s}) = (s_i)_{i \in I}$, which is a vector with $|I|$ components.

\begin{proposition}\label{prop2}
Let $f : \mathbf{F}_q^L \rightarrow C_1/C_2$
be as above.
Define
\begin{eqnarray}
C'_1 &=& \{ (\vec{x}, P_{\overline{I}}(\vec{s})) \mid
\vec{s}\in \mathbf{F}_q^L, \vec{x}\in  f(\vec{s})\},\\
\label{eq:cp1}
C'_2 &=& \{ (\vec{x}, P_{\overline{I}}(\vec{s})) \mid
\vec{s}\in \mathbf{F}_q^L, P_I(\vec{s})=\vec{0},
\vec{x}\in  f(\vec{s})\}.
\label{eq:cp3}
\end{eqnarray}
Then the quantum ramp SS constructed from
$C_1 \supsetneq C_2$ is strongly secure with respect to $I$ and $J$
if and only if
\begin{eqnarray}
\dim P_J (C'_1) - \dim P_J (C'_2)
 &=& 0, \label{eq1001}\\
\dim P_{\overline{J}\cup \{n+1, \ldots, n+|\overline{I}|\}} (C'_1) - \dim P_{\overline{J}\cup \{n+1, \ldots, n+|\overline{I}|\}} (C'_2) &=& |I|. \label{eq1010}
\end{eqnarray}
\end{proposition}
\begin{proof}
By reordering indices we may assume
$I = \{1$, \ldots, $|I|\}$.
For $\vec{s}_I \in \mathbf{F}_q^{|I|}$
define
\[
f'(\vec{s}_I) = \{ (\vec{x}, \vec{s}_{\overline{I}}) \mid
\vec{s}_{\overline{I}}\in \mathbf{F}_q^{|_{\overline{I}}|}, 
\vec{x}\in  f(\vec{s}_I\vec{s}_{\overline{I}})\}.
\]
We have $\dim C'_1 = \dim C_1$, $\dim C'_2 = \dim C_2 +
|\overline{I}|$,
and $f'$ is an $\mathbf{F}_q$-linear isomorphism
from $\mathbf{F}_q^{|I|}$ to $C'_1 / C'_2$.
In the definition of strong security,
the quantum secret has the form
\[
\left(\sum_{\vec{s}_I \in \mathbf{F}_q^{|I|}} \alpha(\vec{s}_I) \ket{\vec{s}_I}\right)
\left(\sum_{\vec{s}_I \in \mathbf{F}_q^{|I|}} \alpha(\vec{s}_I) \bra{\vec{s}_I}\right)
\otimes 
\frac{1}{q^{|\overline{I}|}}
\sum_{\vec{s}_{\overline{I}} \in \mathbf{F}_q^{|\overline{I}|}}  \ket{\vec{s}_{\overline{I}}}\bra{\vec{s}_{\overline{I}}},
\]
whose purification is
\begin{equation}
\sum_{\vec{s}_I \in \mathbf{F}_q^{|I|}} \alpha(\vec{s}_I) \ket{\vec{s}_I}
\otimes 
\frac{1}{\sqrt{q^{|\overline{I}|}}}
\sum_{\vec{s}_{\overline{I}} \in \mathbf{F}_q^{|\overline{I}|}}  \ket{\vec{s}_{\overline{I}}}
\ket{\vec{s}_{\overline{I}}}_R, \label{eq:qs}
\end{equation}
where $\ket{\vec{s}_{\overline{I}}}_R$ is a state vector in the reference system
for purification.
The encoding procedure defined in this appendix transforms
(\ref{eq:qs}) to
\begin{eqnarray}
&& \frac{1}{\sqrt{q^{|\overline{I}|}}} 
\sum_{\vec{s}_I \in \mathbf{F}_q^{|I|}} \alpha(\vec{s}_I)
\frac{1}{\sqrt{|C_2|}}
\sum_{\vec{s}_{\overline{I}} \in \mathbf{F}_q^{|\overline{I}|}}
\sum_{\vec{x}\in f(\vec{s}_I \vec{s}_{\overline{I}})} \ket{\vec{x}}\ket{\vec{s}_{\overline{I}}}_R\nonumber\\
&=& \frac{1}{\sqrt{|C'_2|}}
\sum_{\vec{s}_I \in \mathbf{F}_q^{|I|}} \alpha(\vec{s}_I)
\sum_{\vec{y} \in f'(\vec{s}_I)} \ket{\vec{y}}. \nonumber
\end{eqnarray}
The joint
quantum state of shares and the reference system for purification
can be regarded as encoded shares from
the quantum secret
\[
\sum_{\vec{s}_I \in \mathbf{F}_q^{|I|}} \alpha(\vec{s}_I) \ket{\vec{s}_I},
\]
by using $C'_1/C'_2$ and $f'$.
Equations (\ref{eq1001}) and (\ref{eq1010}) 
is the necessary and sufficient condition \cite{matsumoto14qss}
for
$J$ to be a forbidden set, which shows the theorem. \qed
\end{proof}

We start our proof of Theorem \ref{th1}.
Hereafter we make a different
 assumption $\overline{I} = \{1$, \ldots, $|\overline{I}|\}$.
$C'_1$ and $C'_2$ in Proposition \ref{prop2} become
\begin{eqnarray*}
C'_1 &=& \{ (f(P_1), \ldots, f(P_n), f(Q_1), \ldots, f(Q_{|\overline{I}|})) \mid f \in \mathcal{L}(G) \},\\
C'_2 &=& \{ (f(P_1), \ldots, f(P_n), f(Q_1), \ldots, f(Q_{|\overline{I}|})) \mid f \in \mathcal{L}(G), \forall i\in I, f(Q_i) = 0 \}\\
&=& \{ (f(P_1), \ldots, f(P_n), f(Q_1), \ldots, f(Q_{|\overline{I}|})) \mid f \in \mathcal{L}(G-\sum_{i\in I}Q_i ) \}.
\end{eqnarray*}
Equation (\ref{eq6}) ensures that
\begin{eqnarray}
&&|J| \leq |\overline{I}| + n-1-\deg G \nonumber\\
&\Leftrightarrow & \deg G \leq |\overline{J}| + |\overline{I}| - 1. \label{eq21}
\end{eqnarray}
By \cite{bn:stichtenoth},
(\ref{eq21}) implies that
the mapping $\mathcal{L}(G) \ni
h \mapsto (h(P_{j_1})$, \ldots, $h(P_{j_{|\overline{J}|}})$,
$h(Q_{i_1})$, \ldots, $h(Q_{i_{|\overline{I}|}})) \in P_J (C'_1)$
is $\mathbf{F}_q$-linear and bijective,
where $\{i_1$, \ldots, $i_{|\overline{I}|}\} = \overline{I}$
and $\{j_1$, \ldots, $j_{|\overline{J}|}\} = \overline{J}$.
The above mapping also gives $P_J (C'_2)$
as its image of $\mathcal{L}(G-\sum_{i\in I}Q_i )$.
Equation (\ref{eq3}) implies
\[
|I| = \dim \mathcal{L}(G)- \dim \mathcal{L}(G-\sum_{i\in I}Q_i ),
\]
which in turn implies (\ref{eq1010}) by the above bijection
between $\mathcal{L}(G)$ and $P_J (C'_1)$.

On the other hand, (\ref{eq6}) also ensures that
\begin{eqnarray}
&&|J| \leq |\overline{I}| +  \deg G - L- 2g(F) + 1 \nonumber\\
&\Leftrightarrow & |J| \leq \deg G - |I| - 2g(F) + 1 \nonumber\\
&\Leftrightarrow & |J| \leq \deg (G-\sum_{i\in I}Q_i) - 2g(F) + 1. \label{eq22}
\end{eqnarray}
By \cite{bn:stichtenoth},
(\ref{eq22}) implies that
the mapping $\mathcal{L}(G-\sum_{i\in I}Q_i) \ni
h \mapsto (h(P_{j'_1})$, \ldots, $h(P_{j'_{|J|}}))
\in \mathbf{F}_q^{|J|}$
is $\mathbf{F}_q$-linear and surjective,
where $\{j'_1$, \ldots, $j'_{|J|}\} = J$.
On the other hand, the image of the above mapping
is $P_J(C'_2)$, which means that $P_J(C'_2) = \mathbf{F}_q^{|J|}
= P_J(C'_1)$, which in turn means that (\ref{eq1001}) holds.
Since we have confirmed (\ref{eq1001}) and
(\ref{eq1010}), the proof of Theorem \ref{th1} is
completed by using Proposition \ref{prop2}.
\qed


\end{document}